\documentclass[10pt,a4paper,twocolumn,conference]{IEEEtran}
\usepackage[margin=1.9 cm]{geometry}
\usepackage{cite}	
\usepackage{graphicx} 
\usepackage[latin1]{inputenc} 
\usepackage[T1]{fontenc} 
\usepackage{amsmath,amsfonts,amsbsy,amssymb} 
\usepackage{mathabx} 
\usepackage{amssymb} 
\usepackage{amsmath}
\usepackage{amsthm}	
\usepackage{mathrsfs}
\usepackage[nolist]{acronym} 
\usepackage{tabularx} 
\usepackage{multirow}
\usepackage{wasysym}
\usepackage{float}
\usepackage{url}
\usepackage{color} 
\usepackage{graphicx}
\usepackage[caption=false]{subfig}
\usepackage{lipsum} 
\usepackage{enumitem} 
\usepackage[font=small,skip=0pt]{caption}
\usepackage{caption}
\DeclareCaptionLabelFormat{mylabel}{#1#2.\hspace{1ex}}
\newtheorem{proposition}{Proposition}
\newtheorem{lemma}{Lemma}

\newtheorem{corollary}{Corollary}

\begin{document}

\title{Ultra-Reliable Short Message Cooperative Relaying Protocols under Nakagami-$m$ Fading}
\author{Parisa Nouri\IEEEauthorrefmark{0}, Hirley Alves\IEEEauthorrefmark{0}, Richard Demo Souza\IEEEauthorrefmark{1}, and Matti Latva-aho \IEEEauthorrefmark{0}\\
	\IEEEauthorblockA{
		\IEEEauthorrefmark{0}Centre for Wireless Communications (CWC), University of Oulu, Finland\\
	}
	\IEEEauthorblockA{
		\IEEEauthorrefmark{1} Federal University of Santa Catarina (UFSC), Brazil\\
	}

		firstname.lastname@oulu.fi,
		\IEEEauthorrefmark{1}richard.demo@ufsc.br
}
%
\maketitle

\begin{abstract}
	In the next few years, the development of wireless communication systems propel the world into a fully connected society where the Machine-type Communications (MTC) plays a substantial role as key enabler in the future cellular systems. MTC is categorized into mMTC and uMTC, where  mMTC provides the connectivity to massive number of devices while uMTC is related to low latency and ultra-high reliability of the wireless communications. This paper studies uMTC with incremental relaying technique, where the source and relay collaborate to transfer the message to a destination.
	In this paper, we compare the performance of two distinct cooperative relaying protocols with the direct transmission under the finite blocklength (FB) regime. We define the overall outage probability in each relaying scenario, supposing Nakagami-$m$ fading. We show that cooperative communication outperforms direct transmission under the FB regime. In addition, we examine the impact of fading severity and power allocation factor on the outage probability and the minimum delay required to meet the ultra-reliable communication requirements. Moreover, we provide the outage probability in closed form.
\end{abstract}

\section{Introduction}
The fifth generation (5G) of the cellular systems, which will be introduced in the early 2020s, outperforms the previous generations in terms of data rates and capacity, in addition to supporting new communication protocols to deal with demanding requirements such as latency, reliability and efficiency in MTC~\cite{tullberg2016metis}. In MTC or machine-to-machine communications (M2M), data is automatically exchanged between two MTC nodes or an MTC node and a server where the human cooperation is minimized~\cite{shariatmadari2015machine}. 
5G system outperforms current networks in terms of spectrum efficiency via decreasing the latency and guaranteeing reliability up to $99.999 \%$. Authors in~\cite{andrews2014will} discuss about the 5G technology in more details in terms of requirements and challenges such as data rate, latency, energy and cost issues. 
%
\subsection{Machine-Type Communication}
In recent years, MTC has  received much attention due to the vast applications in wireless communications via supporting transmission/reception of short data packets in comparison with the current communication systems which carry long data packets, and also being a cost-effective and energy efficient technology~\cite{lee2016packet}. 
MTC technology facilitates wide range of applications and is characterized as massive MTC (mMTC) and ultra-reliable MTC (uMTC)~\cite{bockelmann2016massive}. In mMTC, where the blocklength is short, a huge number of devices in a certain domain are covered with low-rate and low-power connectivity, and also considerable reliability in order to cover critical situations, e.g. in sensor networks, smart meters, actuators, etc,~\cite{7041045},~\cite{singh2016selective}. In addition, timing constraints from few seconds to even extremely low end-to-end deadlines in particular applications, are critical concerns in MTC~\cite{biral2016impact}. In uMTC, the connection is supported by transferring the short data packets with ultra-high reliability and low latency, in the scope of less than a millisecond which can be a requirement in several applications such as cloud connectivity, road safety, industrial control and safe interconnection between vehicles~\cite{bockelmann2016massive},~\cite{7529226}. Therefore, owing to the vast applicability of MTC with short packets in cellular network infrastructure, covering a novel  wireless mode, namely ultra-high reliable communication (URC) is a critical concern in 5G~\cite{schotten2014availability}. 
Supporting the ultra-high reliability and low latency are crucial requirements in the upcoming networks in order to improve the security, functionality and the ability of the interaction between different types of communication units such as human-to-human, human-to-machine or machine-to-machine which creates novel business models and applications~\cite{tullberg2016metis}.

In MTC with short packet transmissions, since the length of metadata and information bits are similar, an unsuccessful encoding of the metadata increases the performance loss in the wireless communication systems, particularly under the FB regime~\cite{7529226}. There are several works which have studied different aspects of FB coding. For instance, authors in~\cite{5452208}, provide a tight approximation of achievable coding rate for a specified outage probability under the FB transmissions since majority of the theoretical results assume infinite blocklength (IFB) codes. 
In~\cite{iscan2016comparison}, authors provide an overview of some coding schemes for FB which may be used in 5G. They indicate that the performance of wireless communications with short data packets considerably enhances via novel coding schemes with better minimum distance between the codewords; albeit, the decoders become computationally more sophisticated. Furthermore, authors in~\cite{park2012new}, study the impact of FB on the attainable coding rate. They propose a new power allocation strategy, called modified water-filling, over the AWGN channels. 
They succeed to maximize the lower bound of the achievable coding rate in comparison to the conventional water-filling technique in FB regime.
\subsection{Cooperative Relaying with Finite Blocklength}
In recent years, relaying has become an interesting and challenging research topic. The transmission from the source to the destination is improved with help of intermediate auxiliary nodes. Moreover, relaying is able to combat the wireless fading generated due to the multipath propagation through taking advantage of spatial diversity~\cite{zimmermann2005performance}. The most usual relaying protocol is decode-and-forward (DF), where the relay decodes, encodes and retransmits the message~\cite{zimmermann2005performance}.

Authors in~\cite{hu2015performance} investigate the performance of a two-phase relaying model under the FB and IFB regimes. They consider perfect CSI at the destination where the channel gains of the relaying link and the direct link are combined. They propose different schemes based on the different error scenarios where relaying has a performance advantage over the direct transmission (DT) under the FB regime, particularly when the blocklength is small. Authors in~\cite{swamy2015cooperative} propose cooperative relaying as a way to meet the high reliability and latency requirements.  
In~\cite{7569613}, authors examine the performance of a multi-relay DF protocol with FB under the perfect CSI and partial CSI. 
They indicate that with perfect CSI, the throughput of FB coding is higher than the throughput of IFB coding.
Furthermore, in our previous work~\cite{parisa2017}, we consider cooperative relaying scenarios with perfect CSI for Rayleigh fading channels. We illustrate how the reliability improves via relaying and how we can meet the URC requirements. We examine the probability of successful transmission as a function of the number of information bits and coded blocklength. Moreover, we provide an approximation to the outage probability in closed form. We show that relaying consumes less transmit power compared to DT in order to enable URC with FB coding.
\subsection{Our Contribution}
In this paper, we further investigate the reliability of relaying under the FB regime. We study protocols, namely selection combining (SC) and maximum ratio combining (MRC), where relaying outperforms DT. Furthermore, we illustrate the superiority of MRC over SC in terms of power consumption, latency and reliability under the FB regime and, different from~\cite{hu2015performance}, we present closed form expressions for the outage probability.
The following are considered the contributions of this paper.
\begin{itemize} 
	\item We provide the general expression of the outage probability for each relaying protocol considered in this work and incremental decode-and-forward relaying under general Nakagami-$m$ fading. Therefore, we generalize our previous work in~\cite{parisa2017}.
	\item We extend the work in~\cite{hu2016blocklength} by comparing the MRC scenario with SC relaying. Also, a generalized closed-form expression for the outage probability is attained, thus extending \cite{6888474}.
	\item We show the trade-off between information payload and blocklength under the UR region as a function of the quality of the link, which is reflected by the $m$ parameter of the Nakagami-$m$ distribution. We examine the minimum delay in relaying schemes required to perform in the UR region.
\end{itemize}

{\textbf{Notation}:} Throughout this paper, $f_{W}(\cdot)$ and $F_{W}(\cdot)$ are the probability density function (PDF) and cumulative distribution function (CDF) of a random variable  (RV) W, respectively. $\operatorname{Q}^{-1}(\cdot)$ represents the inverse of the $\operatorname{Q}$-function which is defined as  $\operatorname{Q(w)} = \int_{w}^{\infty} \tfrac{1}{\sqrt{2\pi}}\operatorname{e}^{-t^2/2}dt$~\cite[\S F.2]{miller2012probability}. The outage probability is denoted by  $\epsilon$ and $\operatorname{E[\cdot]}$ is the expectation. $\Gamma(\cdot)$, $\Gamma(\cdot,\cdot)$ and $\operatorname {\cal P}(\cdot,\cdot)$ are the Gamma function~\cite[\S 6.1.1]{abramowitz1964handbook}, incomplete Gamma function~\cite[\S 6.5.2]{abramowitz1964handbook} and Gamma regularized function~\cite[\S 6.5.1]{abramowitz1964handbook}, respectively. $_1F_{1}(a;b;z)$ denotes the regularized hypergeometric function~\cite[\S 15.1.1]{abramowitz1964handbook}.
\section{System Model} \label{sc:system_model}
Consider a source $S$, a destination $D$ and a relay $R$ as illustrated in Fig.~\ref{fig:System Model}. We consider a normalized distance between $S$ and $D$ as $d=1$m, and that $R$ can move in a straight line between $S$ and $D$, while the distance between $D$ and $R$ is denoted by $\beta$. The links denoted as $X$, $Y$ and $Z$ represent the $S$-$R$, $R$-$D$, and $S$-$D$ links respectively, and each transmission takes $n_i$ channel uses where $i \in \{S, R\}$. This means that $n_S$ channel uses for $S$ and $n_R$ channel uses for $R$, respectively. In this scenario, first $S$ sends data to $D$ and $R$ which is known as broadcasting phase. Thereafter, in the relaying phase, $R$ forwards to $D$ only if it decoded the message correctly~\cite{hu2015performance},~\cite{7569613}. 
The received signals are written as $y_{1}\!=\!
h_{1}x + w_{1}$ and $y_{2}\!=\! h_{2}x + w_{2}$ at $D$ and $R$ respectively, and if $R$ participates, the received signal at $D$ is $y_{3}\!=\! h_{3}x + w_{3}$, where $x$ is the transmitted signal with power $P$ and $w_i$ is the AWGN noise vector with power $N_0\!=\!1$ where $i\in \{X,Y,Z\}$. Nakagami-$m$ quasi static fading channels in the $S$-$D$, $S$-$R$ and $R$-$D$ links are denoted as $h_1$, $h_2$, and $h_3$, respectively.
\begin{figure}[b!]
	\centering
	\includegraphics[width=\columnwidth]{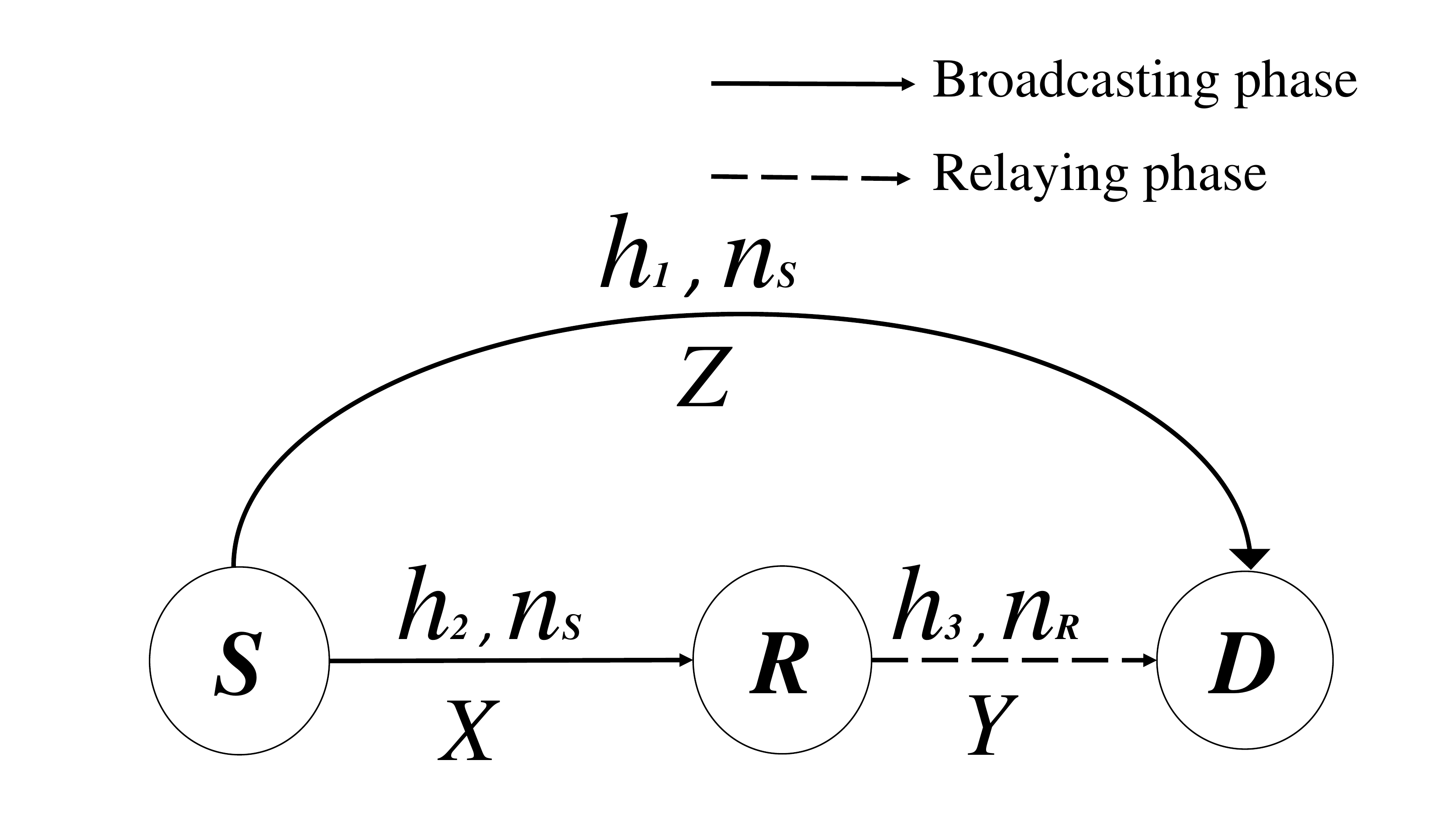}
	\caption{System model for the relaying scenario with a source ($S$), destination ($D$) and a relay ($R$). The links between $S$ to $D$, $S$ to $R$ and $R$ to $D$ are referred as the direct link $h_{1}$, the backhaul link $h_{2}$ and the relaying link $h_{3}$ each of which with $n_{i \in \{S, R\}}$ channel uses respectively.}
	\label{fig:System Model}
\end{figure}
In a DF-based cooperative transmission, the instantaneous SNR for each link depends on the total power constraint $P=P_S+P_R=\eta P+ (1-\eta)P$ which is given by $\Omega_Z \!=\!\eta P|h_{1}|^2/N_{0}$, $\Omega_X \!=\!\eta P|h_{2}|^2/N_{0}$ and $\Omega_Y \!=\! (1-\eta) P|h_{3}|^2/N_{0}$, where $0 <\eta\leqslant1$ and the average SNR in each link is $\gamma_Z =\eta P / N_{0}$, $\gamma_X =\eta P / N_{0}$ and $\gamma_Y = (1-\eta) P / N_{0}$. We consider $\eta$ as the power allocation factor in order to provide a fair comparison between DT and cooperative schemes. Moreover, note that the PDF of Nakagami-$m$ fading is a function of two parameters: the fading parameter $m$ and the scale parameter $\Omega> 0$~\cite{simon2005digital}. The squared envelop is then Gamma distributed as $X \thicksim \mathsf{Gamma}(m_{i} , \Omega_{i} /m_{i} ) $ where $i\in \{X,Y,Z\}$. 

\subsection{Coding Rate and Error Probability at Finite Blocklength } \label{sc:Finite Block length}
In this section,  we revisit the concept of coding rate under the FB regime. In a communication system, first $k$ information bits are mapped to a sequence known as the codeword with a blocklength of $n$ symbols. Thereupon, the generated codeword is sent through the wireless channel. Subsequently, the decoder maps the channel outputs into an estimate of the information bits. Therefore, we can define the coding rate ${\cal R}$ as the ratio of the information bits $k$ to the number of channel symbols $n$ as ${\cal R} \!=\! k/n$~\cite{7529226}. The maximum coding rate ${\cal R}^*(n,\epsilon)$ in bits per channel use (bpcu) is~\cite{7529226}
\begin{equation} \label{eq:maximum rate}
{\cal R}^*(n,\epsilon) = C(\rho) - \sqrt{\frac {V(\rho)}{n}}\operatorname{Q}^{-1}(\epsilon)\log_{2}\operatorname{e},
\end{equation}
where $C$ and $V$ are the positive channel capacity and the channel dispersion, respectively, defined as $C(\rho) = \log_{2}(1+\rho)$ and $V(\rho) = \rho (2+\rho) \big/(1+\rho)^2$, where $\rho$ is the average SNR as $P/N_0$.
For the AWGN channel, $h_{i}=1$ and $\frac{1}{n}\sum_{i}^{n}|x_{i}|^2\leq\rho$ holds~\cite{7529226}. 
%
\begin{figure*}[!t]
	\begin{equation}\label{eq:outage_naka}
	\epsilon\!=\! \dfrac{\mu\bigg[\theta\left(\Gamma\left(m,\dfrac{m\varrho}{\Omega}\right)\!-\!\Gamma\left(m,\dfrac{m\vartheta}{\Omega}\right)\right)\!+\!\Omega \bigg(\Gamma\left(1\!+\!m,\dfrac{m\vartheta}{\Omega}\right)
		\!-\!\Gamma\bigg(1\!+\!m,\dfrac{m\varrho}{\Omega}\bigg)\bigg)\bigg]}{\sqrt{2\pi}}\!+\!\dfrac{\operatorname{\cal P}\left(m,\dfrac{m\vartheta}{\Omega}\right)\!+\!\operatorname{\cal P}\left(m,\dfrac{m\varrho}{\Omega}\right)}{2}.
	\end{equation}
	\hrule
\end{figure*} 

%
%

The outage probability is then the expectation over the instantaneous SNR distribution as follows\footnote{This approximation is accurate for $n$ > 100, as proved for AWGN channels \cite[Figs. 12 and 13]{polyanskiy2010channel}, as well as for fading channels~\cite{yang2014quasi}.}~\cite{hu2016blocklength}
\begin{equation} \label{outage_fading}
\epsilon\approx \operatorname{E}\Bigg[\operatorname{Q}\Bigg(\sqrt{n}\frac{C(\rho|h|^2)-{\cal R}^{*}(n,\epsilon) }{\sqrt{V(\rho|h|^2)}}\Bigg)\Bigg].
\end{equation}
\subsection{Closed-form Expression of the Outage Probability} \label{Q-function}
The outage probability in (\ref{outage_fading}) does not have a closed-form expression, but it can be tightly approximated as we shall see next.
Let us first define $g(x)\!=\!\sqrt{n}\tfrac{C(\rho)-\cal R}{\sqrt{V(\rho)}}$. Then, we resort to a linearization of the Q-function~\cite{7106474},~\cite{6888474}. 
\begin{equation}\label{eq:W(t)}
K(x) \approx \operatorname{Q}(g(x)) =
\left\{
\begin{array}{@{}ll@{}}
\ 1 & x\leqslant\varrho\\
\ \dfrac{1}{2}-\dfrac{\mu}{\sqrt{2\pi}}(x-\theta) & \varrho<x< \vartheta  \;\;\;\;,\\
\ 0 & x\geq \vartheta
\end{array}\right.
\end{equation}
where, $\theta = \tfrac{2^{\cal R}-1}{P}$, $\vartheta=\theta+\sqrt{\tfrac{\pi}{2}\mu^{-2}}$, $\varrho=\theta-\sqrt{\tfrac{\pi}{2}\mu^{-2}}$ and $\mu\!=\!\sqrt{\frac{n}{2\pi}}(e^{2\cal R}-1)^{-\tfrac{1}{2}}$.
Therefore, the outage probability in (\ref{outage_fading}) is reformulated as
\begin{equation}\label{outage_fading_linear}
\epsilon= \operatorname{E}[\operatorname{Q}\left(g(x)\right)] = \int_{0}^{\infty}K(x)f_{X}(x)dx,
\end{equation}
where $f_{X}(x)$ represents the PDF of the SNR of link $X$.
\begin{proposition}\label{propose1}
The approximated outage probability of a communication link following Nakagami-$m$ fading is given in closed-form in (\ref{eq:outage_naka}) on top of the next page, where $\varrho$, $\vartheta$, $\mu$ and $\theta$ are defined in (\ref{eq:W(t)}).
\end{proposition} 
\begin{proof}
	See Appendix A.
\end{proof}
\begin{corollary}
	Notice that for $m=1$, (\ref{eq:outage_naka}) reduces to the approximated outage probability for the Rayleigh fading as ~\cite[\S 5]{parisa2017}.
\end{corollary}
\noindent Thus, our proposed approximation in (\ref{eq:outage_naka}) includes~\cite[\S 5]{parisa2017}, which was first reported in~\cite{6888474}, as a special case.
\section{Performance Analysis of Relaying}\label{sec:performance}
In this section, we examine the outage probability in all relaying schemes considered in this work, as well as for the case of DT which is used for comparison purposes. 
\subsection{Cooperative Transmission} \label{sec:Cooperative schemes}
\subsubsection{Selection Combining (SC)}
In this protocol $D$ tries to decode the messages sent by both $S$ and $R$, separately.
The overall outage probability becomes~\cite{5956530},~\cite{alves2012throughput}	
\begin{equation}
\epsilon_{SC}=\epsilon_{Z}\left(\epsilon_{X}+\left(1-\epsilon_{X}\right)\epsilon_{Y}\right),	
\end{equation}
where $\epsilon_{Z}$ , $\epsilon_{X}$ and $\epsilon_{Y}$ are the outage probabilities of the $S$-$D$, $S$-$R$ and $R$-$D$ links, respectively. The outages $\epsilon_{Z}$, $\epsilon_{X}$ and $\epsilon_{Y}$ are equal to (\ref{eq:outage_naka}) with $m$ and $\Omega$ replaced by $m_{Z}$ and $\Omega_{Z}$, $m_{X}$ and $\Omega_{X}$, and $m_{Y}$ and $\Omega_{Y}$ for each link, respectively.
\subsubsection{Maximum Ratio Combining (MRC)}\label{sec:MRC}
In this protocol, the transmissions from the source and from the relay are coherently combined at the receiver. Hence, the instantaneous SNR after the source and relay transmissions is $\Omega_{W} = \Omega_{Z} + \Omega_{Y}$~\cite{5956530},~\cite{alves2012throughput}. The  outage probability is~\cite{alves2012throughput}
\begin{align}\label{eq:MRC_outage}
\epsilon_{MRC} =\epsilon_{Z}\left(\epsilon_{X}\!+\!\left(1\!-\!\epsilon_{X}\right)\frac{\epsilon_{SRD}}{\epsilon_{Z}}\right),
\end{align}
where $\epsilon_{SRD}$ is the outage probability after MRC of the transmissions from the source to the destination. Moreover, the ratio $\tfrac{\epsilon_{SRD}}{\epsilon_{Z}}$ comes from the conditioning of $\epsilon_{SRD}$ on the fact that the transmission from the source to the destination failed.
In order to calculate (\ref{eq:MRC_outage}), first we need to define the PDF of $W$, a Gamma random variable equal to $Z + Y$ , the sum of the instantaneous SNRs of the $S$-$D$ and $R$-$D$ links, as follows.
%
PDF of random variable $W$ is
\begin{equation}\label{pdf_mrc}
\begin{split}
f_{W}(w)&=\operatorname{exp}\left(-\dfrac{w}{\Omega_{y}}\right)w^{-1+m_{z}+m_{y}}\Omega_{z}^{-m_{z}}\Omega_{y}^{-m_{y}}\\ &\times _1F_{1}\bigg(m_{z},m_{z}+m_{y},w\left(\dfrac{1}{\Omega_{y}}-\dfrac{1}{\Omega_{z}}\right)\bigg),
\end{split}
\end{equation}
{\noindent where $\Omega_{Z}$, $\Omega_{Y}$  and $m_{Z}$, $m_{Y}$ are the instantaneous SNRs and the fading severity of the $S$-$D$ and $R$-$D$ links, respectively.} 
\begin{proposition}
The outage probability after combining the source and relay transmissions, $\epsilon_{SRD}$, is 
\begin{equation}\label{mrc_SRD}
\begin{split}	\epsilon_{SRD}&=\int_{0}^{\varrho}f_W{\left(w|m,\Omega\right)}dw+\int_{\varrho}^{\vartheta}\left(\tfrac{1}{2}\!-\!\tfrac{\mu}{\sqrt{2\pi}}(w\!-\!\theta)\right)\\
&\times f_W{\left(w|m,\Omega\right)}dw. 
\end{split}
\end{equation}
where, $\vartheta$ and $\varrho$ are specified in (\ref{eq:W(t)}) 
\begin{proof}
By plugging (\ref{pdf_mrc}) into (\ref{outage_fading_linear}) and multiplying with $K(x)$, we attain (\ref{mrc_SRD}). Solving (\ref{mrc_SRD}) is difficult and it does not have a closed form expression. Hence, we only work with its numerical integration whenever $m>1$.
\end{proof}
\end{proposition}
Since (\ref{mrc_SRD}) does not have a closed form expression for the outage probability under general Nakagami-$m$, we analyze the particular case when $m=1$, which corresponds to Rayleigh fading. Therefore, the PDF of $W$ becomes~\cite[\S 10]{parisa2017}.

\begin{lemma}\label{proposition2}
	For $m=1$, the outage probability, $\epsilon_{SRD}$, is equal to~\cite[\S 6]{parisa2017}. 
\end{lemma}
\begin{proof}
	The proof can be found in~\cite{parisa2017}.	
\end{proof}
\subsection{Direct Transmission (DT)}
In DT, the source sends the data packets directly to the destination. Since the channel coefficients are Nakagami-$m$ distributed, the instantaneous  SNR of DT scheme is $\Gamma_{Z} \thicksim \mathsf{Gamma}(m_{Z}, \Omega_{Z} /m_{Z})$. Then, the outage probability of DT is calculated according to (\ref{eq:outage_naka}), where $m = m_{Z}$, $\Omega_{Z} = P_S/ N_0 $, and is denoted as $\epsilon_{DT}$.
\section{Numerical Results}\label{sc:result}
In this paper we study the performance of incremental relaying protocols under the FB regime for the Nakagami-$m$ fading. First, we compare the cooperative schemes with DT as a function blocklength and transmit power. We show that the cooperative relaying protocols outperform the DT in terms of transmit power and the blocklength under the UR region. 
The accuracy of our analytical model is proved via the numerical results. Unless stated otherwise, we assume $n=500$ where $n_{S}=n_{R}=n$, $k=250$, average SNR as $10$ dB and that $R$ is exactly midway between $S$ and $D$, with $\beta=\tfrac{1}{2}$.
\begin{figure}[!t]
	\centering
	\includegraphics[width=\columnwidth]{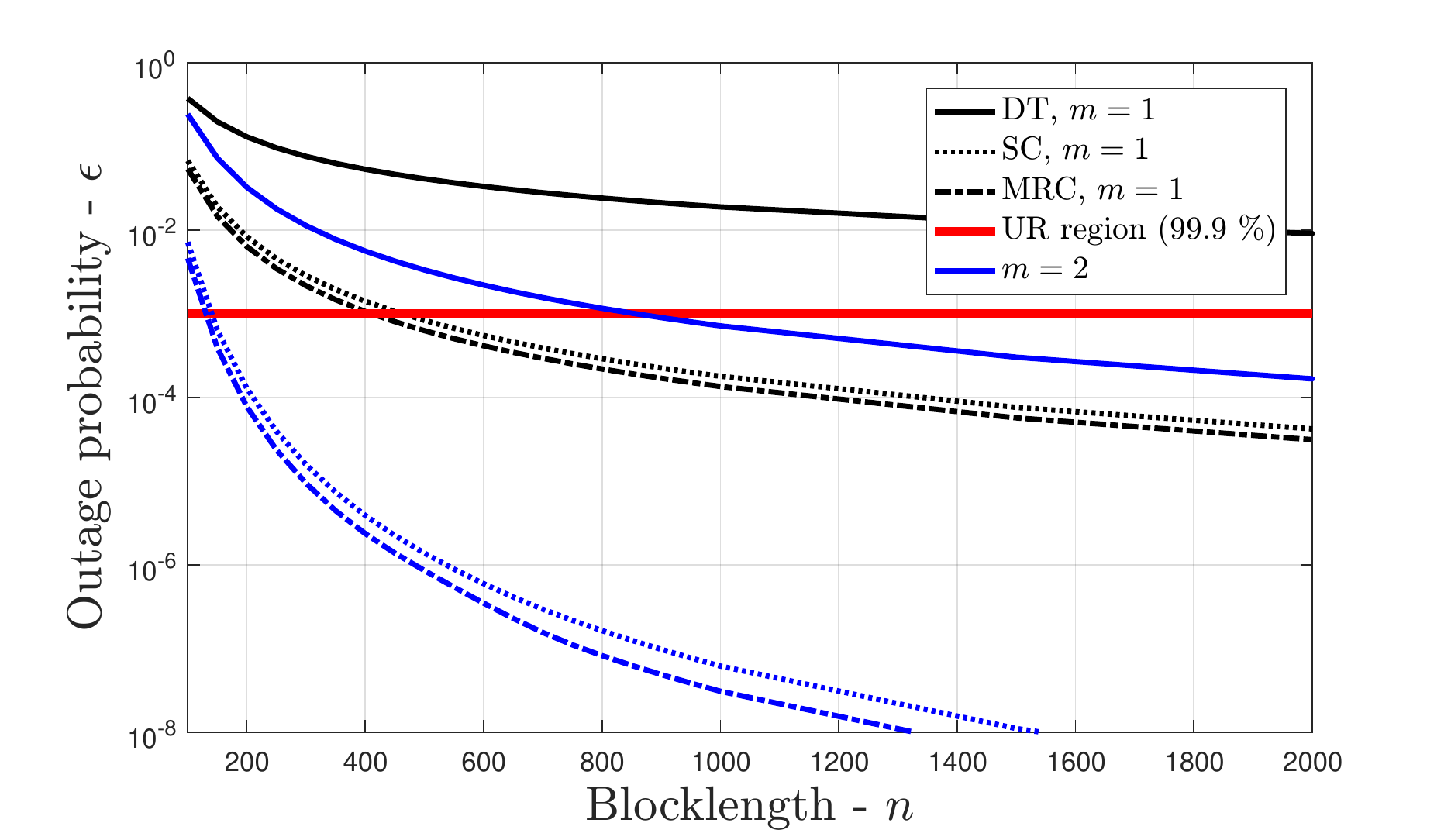}
	\vspace{-2mm}
	\caption{Outage probabilities for DT, SC, MRC with $k=250$ and average SNR as 10 dB.}\label{fig:outageVSblocklength}
\end{figure}
\subsection{On the impact of reliability improvement}
Fig.\ref{fig:outageVSblocklength} compares two distinct relaying protocols with DT in terms of the overall outage probability and fading severity. It can be clearly seen that MRC and SC outperform DT and perform closely in the entire range. In addition, higher values of $m$ results in lower outage
probability due to the improvement of LOS. Therefore, the availability of LOS should be taken into account when designing and deploying networks with stringent reliability requirements under the FB regime.
\begin{figure}[!htp]
	\centering
	\includegraphics[width=\columnwidth]{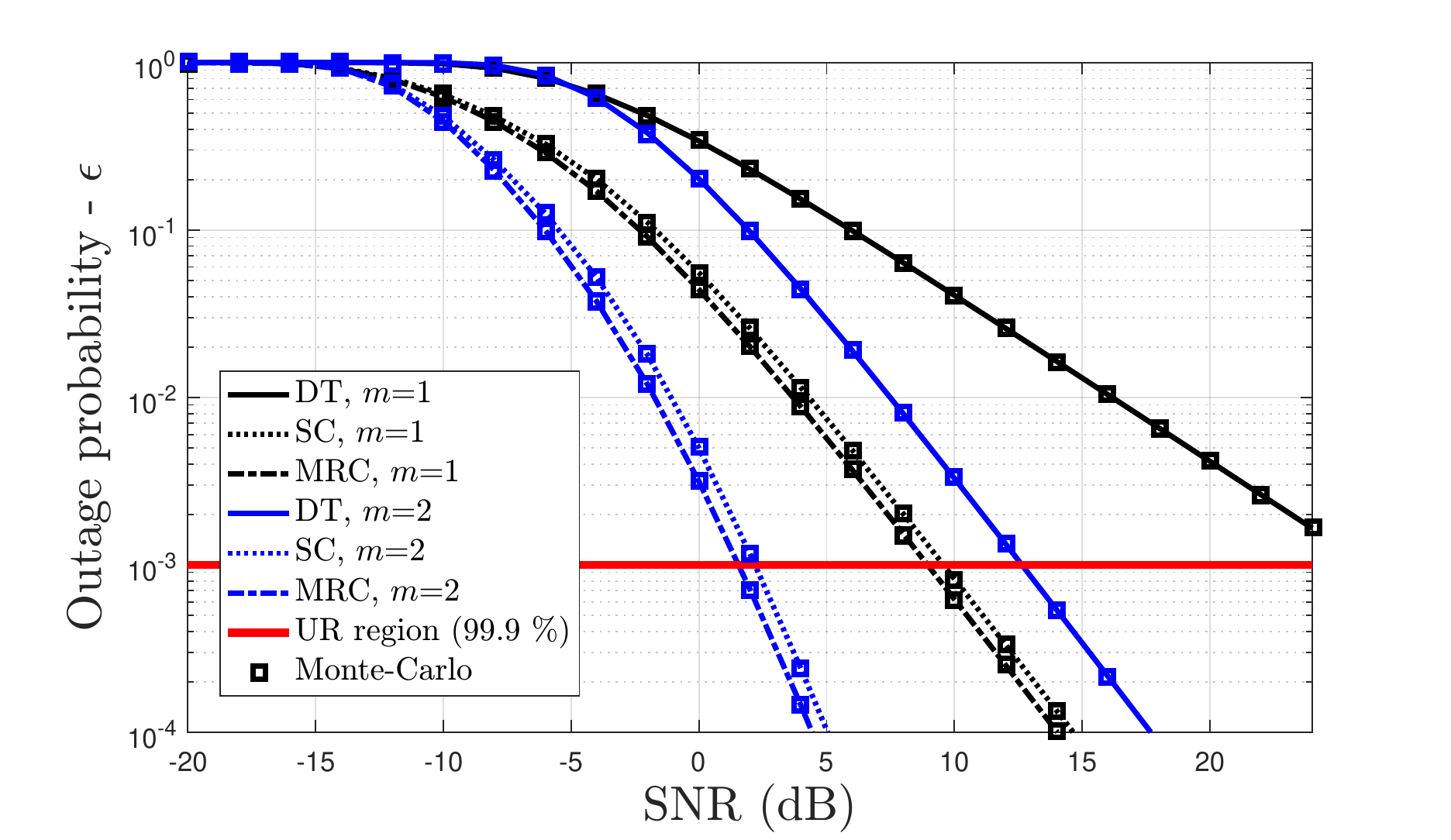}
	\vspace{-2mm}
	\caption{Outage probabilities for DT, SC, MRC with $k=250$ and $n=500$. }\label{fig:outageVSpower}
\end{figure}
Fig.\ref{fig:outageVSpower} shows that the cooperative protocols outperform DT in terms of transmit power consumption. The diversity gain achieved by the relaying schemes is more evident at high SNR and MRC performs better than SC. Hence, URC becomes feasible in case of less severe fading (larger $m$) and/or by
utilizing cooperative schemes. Note also that the analytical and numerical results match very well.
\begin{figure*}[!t]
	\centering
	\includegraphics[width=2\columnwidth,height=1.9in]{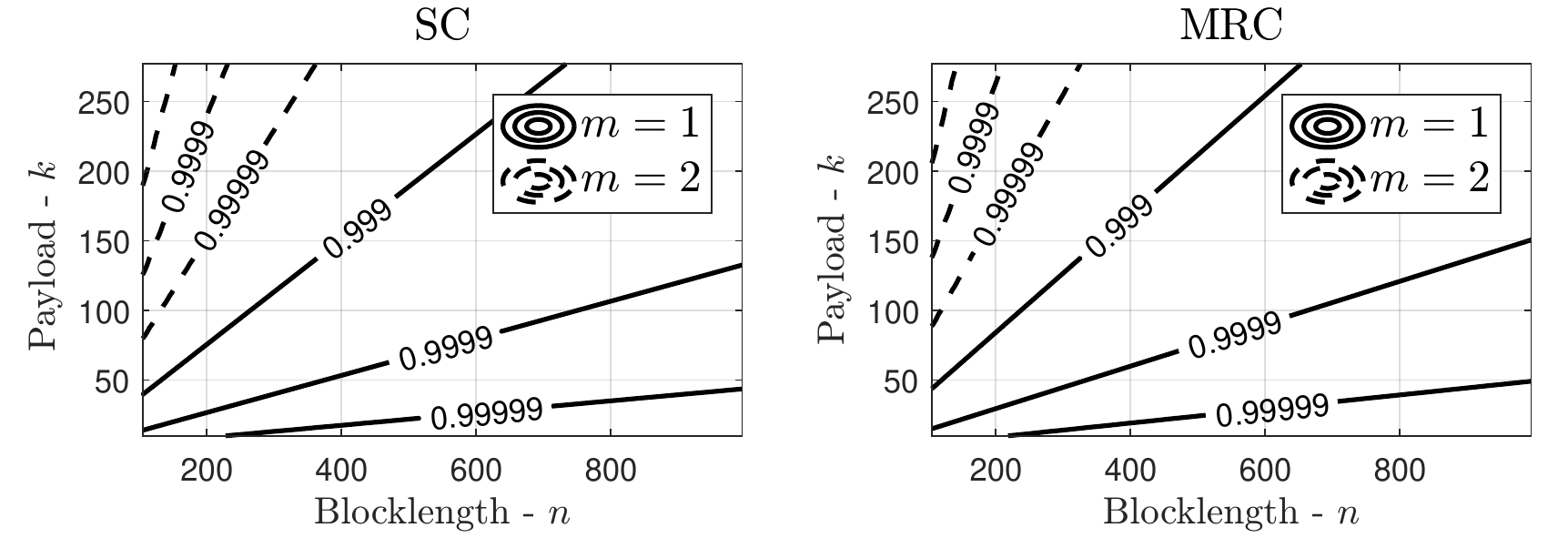}
	\vspace{0mm}
	\caption{Probability of successful transmission in SC and MRC protocols as a function of the blocklength $n$ and payload $k$ with average SNR as 10 dB.}\label{fig:contour}
\end{figure*}
\begin{figure}[!t]
	\centering
	\includegraphics[width=\columnwidth]{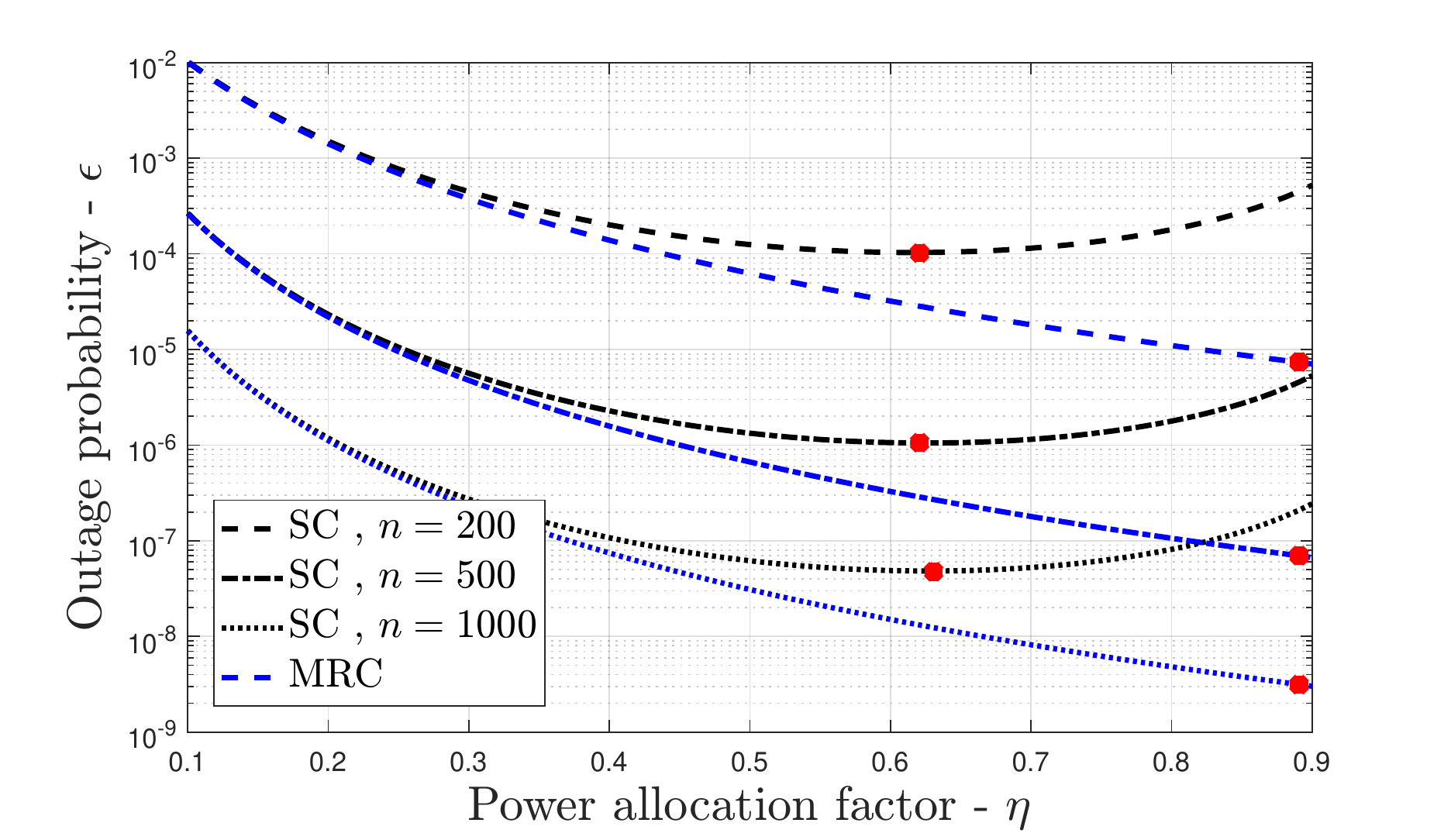}
	\vspace{-2mm}
	\caption{Impact of power allocation factor $\eta$ on the outage probability of SC and MRC with different blocklengths $n$. We consider $m=2$, $k=250$ and average SNR as $10$ dB. The points marked in red are the minimum outage probabilities of SC and MRC.}\label{fig:powerallocation}
\end{figure}
\subsection{URC requirements}
In Fig.\ref{fig:contour}, we examine the impact of coding rate on the performance of SC and MRC protocols with different values of $m$ under the UR region. We illustrate that MRC protocol outperforms SC under th FB regime. For instance, with $m=1$ and $n=300$, MRC provides $99.99 \%$ reliability with $k=45$, while SC supports equal reliability with $k=40$. Hence, SC is more affected by the coding rate growth compared to the MRC protocol. Moreover, we can clearly see that improving the LOS (larger $m$), considerably improves the reliability in wireless communications; thus, the size of the payload is not a critical concern anymore.

In Fig.\ref{fig:powerallocation}, we show the effect of power allocation factor on the outage probability. The optimal value of $\eta$, where the outage probability is minimized with different blocklengths, is approximately $0.62$ and $0.9$ for SC and MRC, respectively. From the figure, we see that in good LOS conditions more power is allocated to $S$, though $R$ is needed to provide the diversity order with respect to DT and thus achieve the required reliability target. 

In the following figure, we compare the delay in relaying schemes with equal power allocation $\eta =0.5$ and optimal power allocation $\eta=0.62$ and $\eta=0.9$ for SC and MRC, respectively. The choice of the minimum delay $\delta$ is in such a way that minimizes the outage probability constrained to a specific interval of interest which gives the optimal values of $n_S$ and $n_R$, and is a nonlinear optimization problem. The optimization problem numerically is solved via the Matlab function $fmincon$\footnote{Interior point algorithm is used to solve the nonlinear optimization problem~\cite{waltz2006interior}.}. Delay $\delta$ is equal to the symbol time $T_s$ multiplied with blocklength $n$ plus the number of retransmissions. For example, in Fig.\ref{fig:totallatency}, where $\eta=0.5$, future releases of LTE foresee a minimum symbol period of $T_s = 8.33\mu$s~\cite{yilmazultra}  and regardless of the retransmissions, $99.99 \%$ reliability is feasible via MRC protocol with $\delta = 3.5 $ms while latency reduces to $\delta = 1.66$ms with $99.9 \%$ reliability. On the other hand, SC protocol provides $99.99 \%$ reliability with $\delta = 3.9$ms when $n_S = n_R$. Note that DT would not be able to cope with such stringent latency requirements. Therefore, we should consider a good trade-off between the reliability and latency in cooperative communication protocols. 
Moreover, we indicate that applying the optimal values of $\eta$ shown in Fig.\ref{fig:powerallocation}, slightly reduce the latency in relaying schemes compared to the equal power allocation to $S$ and $R$ under the same assumptions; however, behavior of the protocols alters under the optimal $\eta$. Moreover, in SC with $n_S \neq n_R$ and $\eta = 0.5$, latency reduces in comparison to the case when $n_S = n_R$. For instance, when $n_s=227$ and $n_R=219$,  $99.99 \%$ reliability is feasible with $\delta=3.7$ms compared to the case when $n_S = n_R$. Hence, different parameters such as LOS, transmit power, coding rate of each link, delay and reliability should be all considered in designing the networks to support the URC requirements.
\section{Conclusions and Final Remarks} \label{sc:conc}
Herein we assess the relay communication under the finite blocklength regime. The performance of two relaying schemes, SC and MRC, are compared to the DT in terms of the outage probability under the Nakagami-$m$ fading distribution. We indicate that the relaying improves the communication efficiency which is more obvious at high SNR regime through providing higher order of diversity and consuming less transmit power. Moreover, we show the optimal number of channel uses for each link for all relaying schemes considered in this work in order to meet the latency constraint under the UR region. We show that lower coding rates, improve the reliability up to $99.999 \%$.
\begin{figure}[!t]
	\centering
	\includegraphics[width=\columnwidth, height =2in]{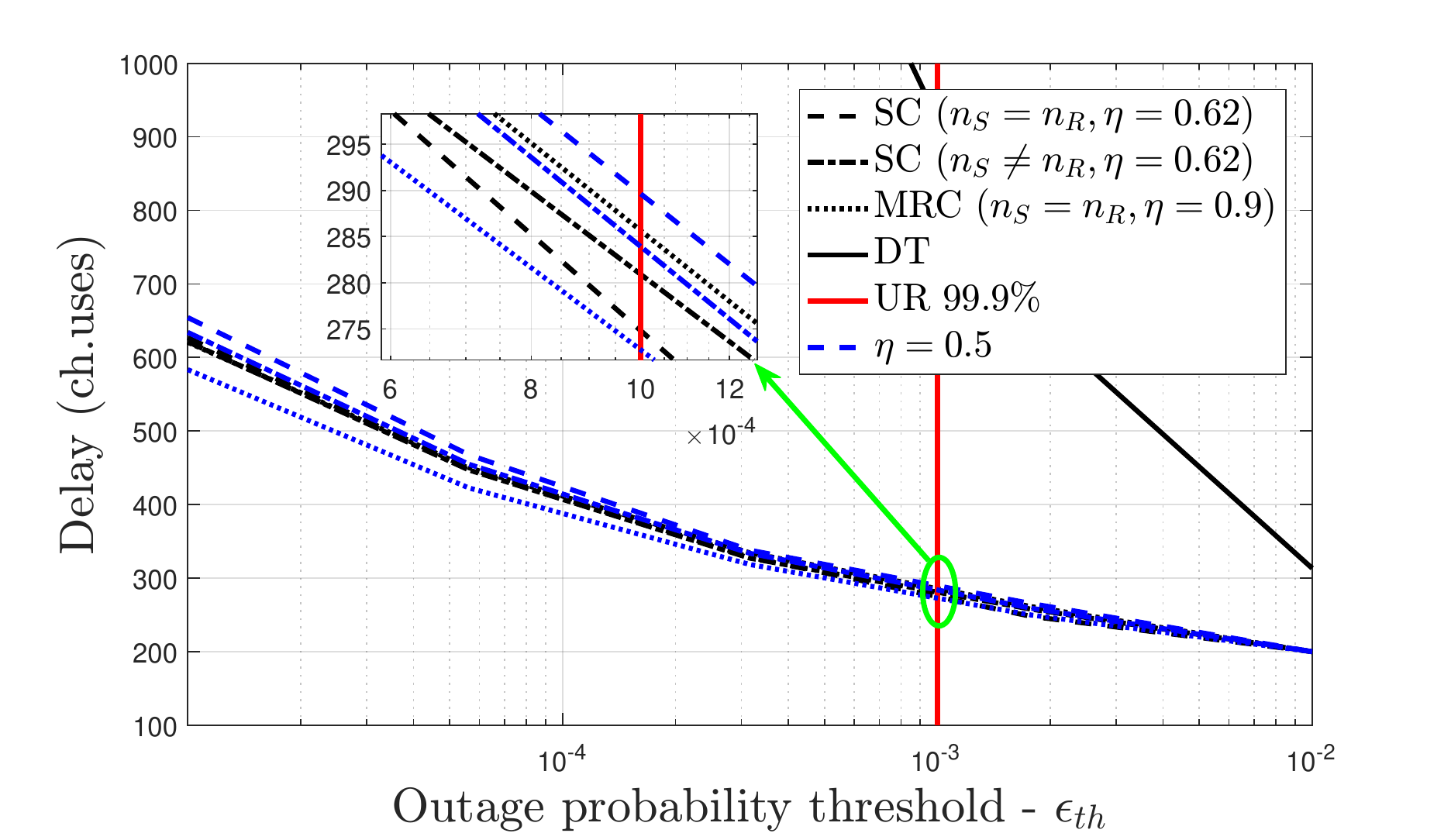}
	\vspace{-4mm}
	\caption{Total delay in terms of channel uses with equal power allocation ($\eta=0.5$) and $\eta=0.62$ and $\eta=0.9$ for SC and MRC, respectively with $m=2$, $k=250$ and average SNR as $10$ dB.}\label{fig:totallatency}
\end{figure}
Furthermore, we indicate that increasing the LOS, increases the reliability while the transmit power is kept low. In addition, we show that if $S$ and $R$  transmit with different coding rates in SC scenario, the latency reduces in comparison to the case of having equal coding rates for both direct and relay links.
We also prove the correctness of our analytical model via the numerical results.
All in all, we conclude that the relaying is a desirable technique in order to improve the system quality, particularly in M2M communications where low latency is a concern with short packet transmissions. In addition, relaying requires less transmit power to work under the UR region and increasing the LOS improves the reliability. More
over, MRC outperforms SC and DT in terms of power consumption and latency requirements. Finally, in our future work, we will work on the cooperative relaying under the imperfect CSI with FB coding.
\begin{figure*}[!t]
	\begin{equation}\label{eq:Outage_overall}
	\begin{split}
	\epsilon&\!=\!\dfrac{F_X{\left(\vartheta,m,\Omega\right)}\!+\!F_X{\left(\varrho,m,\Omega\right)}}{2}\!+\!\dfrac{\mu\theta\bigg(F_X{\left(\vartheta,m,\Omega\right)}\!-\!F_X{\left(\varrho,m,\Omega\right)}\bigg)}{\sqrt{2\pi}}\!-\!\dfrac{\mu}{\sqrt{2\pi}}\bigg(\bigg(\bigg(\dfrac{1}{\Omega}\bigg)^{-m} \Omega^{1\!-\!m}\bigg(\Gamma\left(1\!+\!m,\dfrac{m\varrho}{\Omega}\right)\\
	&\!-\!\Gamma\left(1\!+\!m, \dfrac{m\vartheta}{\Omega}\right)\bigg)\bigg)\bigg/ \Gamma\left(1\!+\!m\right)\bigg)\!=\!\dfrac{1}{2}\bigg(\operatorname{\cal P}\left(m,\dfrac{m\vartheta}{\Omega}\right)\!+\!\operatorname{\cal P}\left(m,\dfrac{m\varrho}{\Omega}\right)\bigg)\!+\!\dfrac{\mu m \theta}{\sqrt{2\pi}\Gamma\left(1\!+\!m\right)}\bigg(\Gamma\left(m,\dfrac{m\varrho}{\Omega}\right)\\
	&\!-\!\Gamma\left(m,\dfrac{m\vartheta}{\Omega}\right)\bigg)\!+\!\dfrac{\mu\Omega}{\sqrt{2\pi}\Gamma(1\!+\!m)}\bigg(\Gamma\left(1\!+\!m,\dfrac{m\vartheta}{\Omega}\right)\!-\!\Gamma\left(1\!+\!m,\dfrac{m\varrho}{\Omega}\right)\bigg)\!=\!\dfrac{\mu}{\sqrt{2\pi}}\bigg[\theta\bigg(\Gamma\left(m,\dfrac{m\varrho}{\Omega}\right)\!-\!\Gamma\left(m,\dfrac{m\vartheta}{\Omega}\right)\bigg)\\
	&\!+\!\Omega \bigg(\Gamma\left(1\!+\!m,\dfrac{m\vartheta}{\Omega}\right)\!-\!\Gamma\left(1\!+\!m,\dfrac{m\varrho}{\Omega}\right)\bigg)\bigg]\!+\!\dfrac{1}{2}\bigg(\operatorname{\cal P}\left(m,\dfrac{m\vartheta}{\Omega}\right)\!+\!\operatorname{\cal P}\left(m,\dfrac{m\varrho}{\Omega}\right)\bigg).\\
	\end{split}	
	\end{equation}
	\hrule
\end{figure*}
\section*{Appendix A}
\section*{Proof of the Proposition~\ref{propose1}}
\begin{proof}
	Let the PDF of the instantaneous SNR of the link denoted by the random variable $X \thicksim \mathsf{Gamma}(m,\Omega/m)$ be 
	\begin{align}\label{eq:pdf_gamma}
	f_X{(x|m,\tfrac{\Omega}{m})} =(\frac{\Omega}{m})^{-m}\frac{\exp(-\frac{mx}{\Omega})x^{m-1}}{\Gamma(m)} &\quad x>0
	\end{align}
	then, plugging (\ref{eq:W(t)}) and (\ref{eq:pdf_gamma}) into (\ref{outage_fading_linear}), we attain the outage probability equal to 
	\begin{equation}
	\epsilon\!=\! \int_{0}^{\varrho}f_X{\left(x|m,\Omega\right)}dx\!+\!\int_{\varrho}^{\vartheta}\left(\tfrac{1}{2}\!-\!\tfrac{\mu}{\sqrt{2\pi}}(x\!-\!\theta)\right)f_X{\left(x|m,\Omega\right)}dx. 
	\end{equation}
	After some algebraic manipulations we attain (\ref{eq:Outage_overall}) on top of the this page, where  $F_X(x;m,\Omega)\!=\!\operatorname{\cal P}(m,\tfrac{mx}{\Omega})$ and $\tfrac{m\mu\theta\Gamma(m,\tfrac{mx}{\Omega})}{\sqrt{2\pi}\Gamma\left(1\!+\!m\right)}$ can be simplified as $\tfrac{\mu\theta\Gamma\left(m,\tfrac{mx}{\Omega}\right)}{\sqrt{2\pi}\Gamma(m)}$.
\end{proof} 

\section*{Acknowledgments}
This work is partially supported by Aka Project SAFE (Grant no. 303532), and by Finnish Funding Agency for Technology and Innovation (Tekes), Bittium Wireless, Keysight Technologies Finland, Kyynel, MediaTek Wireless, Nokia Solutions and Networks, and CNPq (Brazil).


\end{document}